\documentclass[12pt]{article}
\usepackage{amsmath}
\usepackage{graphicx,psfrag,epsf}
\usepackage{enumerate}
\usepackage{natbib}
\usepackage{url} % not crucial - just used below for the URL 

\usepackage{setspace}
\usepackage[english]{babel}
\usepackage{pgfplots} 
\usepackage{float}
\usepackage{subcaption}
\usepackage{layout}
\usepackage{amsfonts} 
\usepackage{amsmath}
\usepackage{bm}
\usepackage{multirow}
\usepackage{array}
\usepackage{booktabs}
\usepackage{mathtools}
\usepackage{array}
\usepackage[utf8]{inputenc}
\usepackage{graphicx}
\usepackage{amsthm}
\newcommand{\blind}{0}

\theoremstyle{plain}% Theorem-like structures provided by amsthm.sty
\newtheorem{theorem}{Theorem}[section]
\newtheorem{lemma}[theorem]{Lemma}
\newtheorem{corollary}[theorem]{Corollary}

\theoremstyle{definition}
\newtheorem{definition}[theorem]{Definition}

\theoremstyle{remark}

\begin{document}

\def\spacingset#1{\renewcommand{\baselinestretch}%
{#1}\small\normalsize} \spacingset{1}

%%%%%%%%%%%%%%%%%%%%%%%%%%%%%%%%%%%%%%%%%%%%%%%%%%%%%%%%%%%%%%%%%%%%%%%%%%%%%%

\if0\blind
{
  \title{\bf The extension of Pearson's correlation coefficient, measuring noise, and selecting features}
  \author{Reza Salimi\hspace{.2cm}\footnote{Salimireza21@gmail.com}\\
    Kharazmi University of Tehran, Department of Finance\\
    and \\
    Kamran Pakizeh\footnote{k.pakizeh@khu.ac.ir} \\
    Kharazmi University of Tehran, Department of Finance}
  \maketitle
} \fi

\if1\blind
{
  \bigskip
  \bigskip
  \bigskip
  \begin{center}
    {\LARGE\bf The extension of Pearson's correlation coefficient, measuring noise, and selecting features}
\end{center}
  \medskip
} \fi

\bigskip
\begin{abstract}
\setstretch{1.3}
Not a matter of serious contention, Pearson's correlation coefficient is still the most important statistical association measure. Restricted to just two variables, this measure sometimes doesn't live up to users' needs and expectations. Specifically, a multivariable version of the correlation coefficient can greatly contribute to better assessment of the risk in a multi-asset investment portfolio. Needless to say, the correlation coefficient is derived from another concept: covariance. Even though covariance can be extended naturally by its mathematical formula, such an extension is to no use. Making matters worse, the correlation coefficient can never be extended based on its mathematical definition. In this article, we briefly explore random matrix theory to extend the notion of Pearson's correlation coefficient to an arbitrary number of variables. Then, we show that how useful this measure is at gauging noise, thereby selecting features particularly in classification. 
\end{abstract}

\noindent%
{\it Keywords:}  FC-Datasets, FU-Datasets, Maximal Eigenvalue, Euclidean Matrix Norm
\vfill

\newpage
\spacingset{1.45} % DON'T change the spacing!
\section{Introduction}
\label{sec:intro}

In order to devise a better way to deal with noise while going over \cite{LiewMayster}, we delved into the magnificent work of \cite{MarcenkoPastur} and found that the distribution proposed by the authors has a premise that has fallen under the radar. Researchers conducting data-scientific endeavors usually employ Marchenko-Pastur distribution to figure out how noisy their dataset is, subject to study, and denoise it. Taking into consideration a dataset as a matrix, we have columns representatives of variables and rows as samples. In such a dataset, the row is the dimension that grows, and the number of columns fixed. However, the Marchenko-Pastur distribution occurs while both matrix dimensions grow. In order to address this issue, we devised a method, finally resulting in the extended correlation coefficient. Upon this, we tried to gain a better understanding of how to extend the notions of covariance and correlation to multi-variable according to their natural definitions. Based on its true definition that is
$cov_{x,y}=\mathbb{E}[(x-\bar{x})(y-\bar{y})]$,
covariance can be extended to three or more variables: 
$cov_{x,y,z}=\mathbb{E}[(x-\bar{x})(y-\bar{y})(z-\bar{z})]$.

However, such an extension does not show direction, something covariance serves to because the number of variables\textemdash not their distance from their means\textemdash  affects most the direction of variables toward each other. In order to drive the correlation coefficient, division by $\sigma_{x}\sigma_{y}\sigma_{z}$ is useless. It is because $\sigma$ is the second root of the second central momentum\textemdash variance. In fact, we should take the third root of the third central momentum into consideration. However, a significant obstacles exists: the Cauchy-Schwartz inequality, which makes the correlation coefficient lie between -1 and +1, holds just for the second power and not higher ones. Actually,  this inequality gives Pearson's correlation coefficient the power to measure the magnitude of statistical association, what covariance itself could not provide.
While the mentioned inequality would be beneficial for higher orders, the derived correlation coefficient, so-called the extended covariance, becomes impractical due to its inability to indicate direction in the numerator.

In existing literature, the multiple correlation coefficient stands out as a measure for the statistical association of multiple variables. This coefficient, denoted as $R$ for three variables $x$,$y$, and $z$ is calculated as $R=\sqrt{\frac{r_{zx}^2+r_{zy}^2-2r_{zx}r_{zy}r_{xy}}{1-r_{xy}^2}}$ concerning $z$. It reveals how well $z$ can be predicted using a linear function of $x$ and $y$, though it is obviously not independent of variables. Recent advancements, such as the non-linear generalization proposed by \cite{AzadkiaChatterjee}, offer alternatives in measuring conditional dependence through regression. This indicator authentically assesses the dependence between $x$ and $y$ relative to $z_{1}$, $z_{2}$, ..., $z_{n}$, converging to a limit within [0,1]. \cite{Burgin} contributes to financial estimation by aggregating all pairwise correlations into a single measure. Additionally, leveraging canonical correlations, \cite{GoreckiKrzyskoRatajczakWolynski} extends the classical correlation coefficient to functional data, presenting a novel approach in multivariate functional space.
 
In recent years, studies on measures of statistical association have predominantly aimed to address inherent shortcomings, such as the incapacity to detect non-monotonic associations. A comprehensive examination, as conducted by \cite{JosseHolmes}, has yielded numerous suggestions to rectify these weaknesses. Notable among them are the RV coefficient proposed by \cite{Escoufier}, \cite{Gower}, \cite{Schoenberg}, \cite{PaulEscoufier}, \cite{HolmesSusan}, the dCov coefficient introduced by \cite{SzekelyRizzoBakirov}, \cite{Newton}, \cite{Romano}, \cite{Romanoj}, along with kernel-based and graph-based measures. While most of these measures focus on or offer generalized versions for associations between two matrices, recent brilliant approaches, as observed in works by \cite{Chatterjee}, \cite{CaoBickel}, \cite{NabarunGhosalSen}, \cite{MathiasHanShi}, \cite{GhosalSen}, and \cite{HongjianDrtonHan}, provide valuable insights into related developments.

In this paper, our approach to extending Pearson's correlation differs significantly from previous methods. We demonstrate that exploring measures of association within a broader context, such as random matrix theory, yields more robust insights. Random-matrix methods have gained prominence in recent decades for their profound knowledge and versatile techniques. Initially introduced by \cite{Wishart} and later advanced by \cite{Wigner} in the statistical distribution of nuclear energy levels, random matrix theory has become foundational in various scientific domains. He also published two consecutive prominent articles regarding the enormous expansion of symmetric matrices and level spacing distribution of symmetric matrices: \cite{Wignerr} and \cite{Wignerrr}. Freeman Dyson, a pioneering mathematician, viewed random matrices as a new statistical framework, establishing the mathematical foundations in a series of papers on energy statistics \cite{Dyson, Dysonn, Dysonnn}. In contemporary statistics, random matrix theory addresses problems in dimension reduction, hypothesis testing, clustering, regression analysis, and covariance estimation, as highlighted by \cite{DebashisAue}.

Major problems in the study of multivariate data revolve around the eigendecomposition of certain Hermitian or symmetric matrices, falling into the Wishart and double Wishart problem categories. Principal component analysis (PCA) \cite{JolliffeCadima}, factor analysis, MANOVA, CCA \cite{HardleSimar}, tests for equality of covariance matrices, and tests for linear hypotheses in multivariate linear regression problems are exemplars of these categories. Random matrices play a significant role in multivariate linear regression, classification, and clustering problems.

Analyzing the behavior of eigenvalues and eigenvectors of random symmetric or Hermitian matrices dates back to the work of \cite{Pearson} who introduced the concept of dimension reduction through PCA. In this article, we leverage eigenvalues, eigenvectors, and matrix decomposition to enhance our understanding of the correlation coefficient and extend it to multiple variables.

\section{Fully-Correlated and Fully-Uncorrelated}
\label{Fully-Correlated and Fully-Uncorrelated}

In the case of two random variables, positive and negative complete correlations are mapped into +1 and -1, the ends of the interval in which the correlation coefficient lies. A correlation formula confined within a bounded interval is crucial for providing a meaningful assessment of the association magnitude. To establish a mathematical relation for an n-variable version of the correlation coefficient, it is imperative to define the cases of complete correlation and complete uncorrelation, which correspond to the extremities of the interval within which the forthcoming formula is constrained.

\subsection{Fully-Correlated Dataset}
\label{Fully-Correlated Dataset}

\begin{definition}[\textbf{FC-dataset}]
We refer to a dataset wherein all variables have perfect correlation pairwise, regardless of the sign, as the fully-correlated dataset, abbreviated to \textbf{FC-dataset}.
\end{definition}

The following shows FC-datasets for three variables: $A$, $B$, and $C$. 

\par
\begin{center}
\begin{minipage}{0.45\textwidth}
\includegraphics[width=\linewidth]{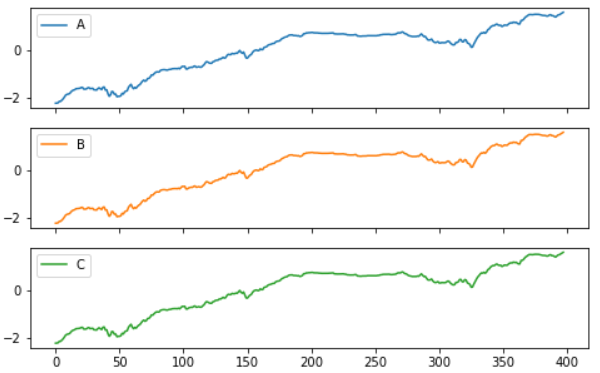}
\captionof*{figure}{(a)}
\end{minipage}\hfill
\begin{minipage}{0.45\textwidth}
\includegraphics[width=\linewidth]{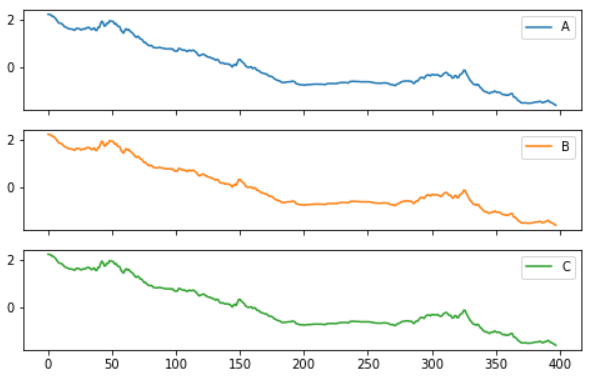} 
\captionof*{figure}{(b)}   
\end{minipage}
\begin{minipage}{0.45\textwidth}
\includegraphics[width=\linewidth]{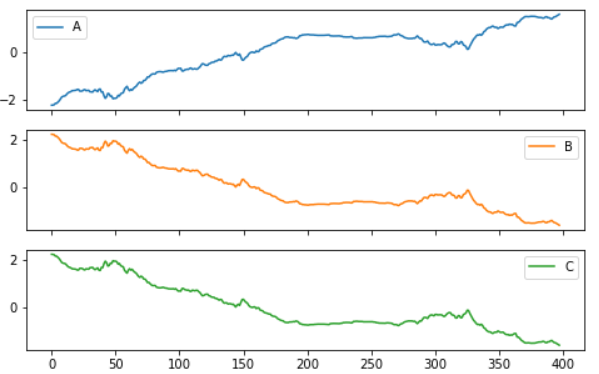}
\captionof*{figure}{(c)}
\end{minipage}\hfill
\begin{minipage}{0.45\textwidth}
\includegraphics[width=\linewidth]{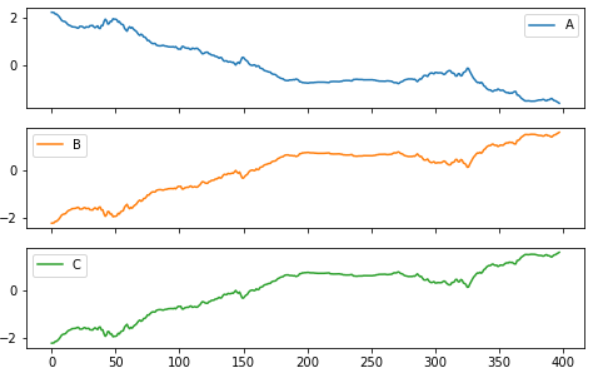}
\captionof*{figure}{(d)}
\end{minipage}\hfill
\begin{minipage}{0.45\textwidth}
\includegraphics[width=\linewidth]{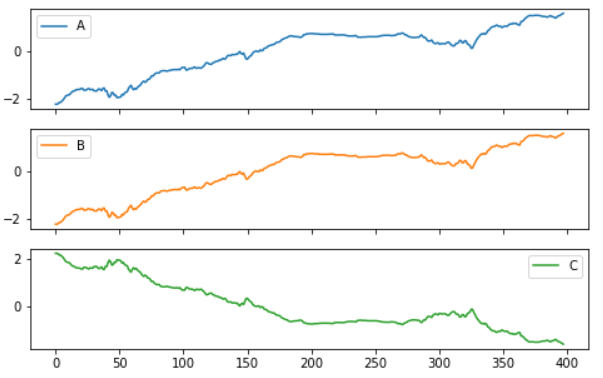}
\captionof*{figure}{(e)}
\end{minipage}\hfill
\begin{minipage}{0.45\textwidth}
\includegraphics[width=\linewidth]{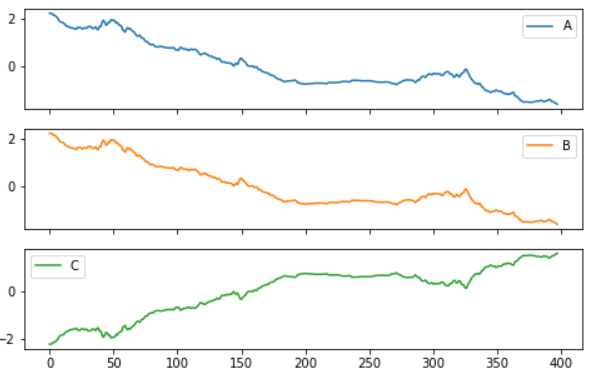}
\captionof*{figure}{(f)}
\end{minipage}\hfill
\begin{minipage}{0.45\textwidth}
\includegraphics[width=\linewidth]{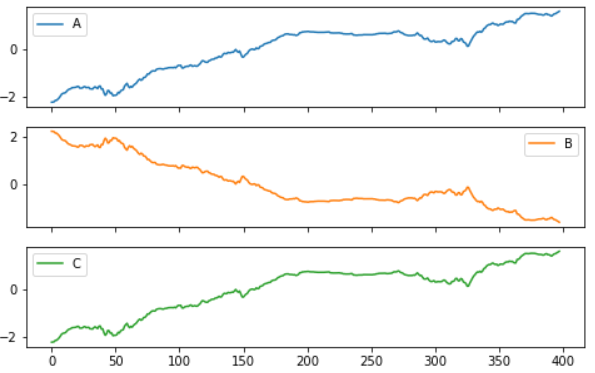}
\captionof*{figure}{(g)}
\end{minipage}\hfill
\begin{minipage}{0.45\textwidth}
\includegraphics[width=\linewidth]{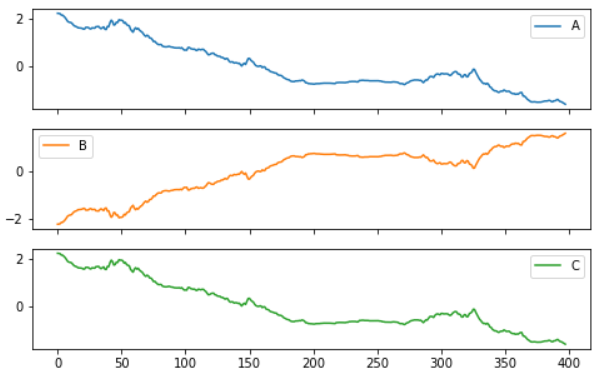}
\captionof*{figure}{(h)}
\end{minipage}\hfill
\captionof{figure}{A Set of Three-variable FC-datasets}
\label{fig1}
\end{center}

Each of these eight three-variable FC-datasets belongs to one of the following correlation matrices:

\[
  \left[\begin{array}{ccc}
+1 & +1 & +1\\
+1 & +1 & +1\\
+1 & +1 & +1
  \end{array}\right] 
  \left[\begin{array}{ccc}
+1 & +1 & -1\\
+1 & +1 & -1\\
-1 & -1 & +1
  \end{array}\right]
  \left[\begin{array}{ccc}
+1 & -1 & +1\\
-1 & +1 & -1\\
+1 & -1 & +1
  \end{array}\right]
  \left[\begin{array}{ccc}
+1 & -1 & -1\\
-1 & +1 & +1\\
-1 & +1 & +1
  \end{array}\right]
\]

We aim to generalize this fact by asserting that for $n$ in $\mathbb{N}$, there exist $2^{n-1}$ FC-datasets. However, before delving into this generalization, it is essential to introduce the concept of \textbf{transitivity} of the correlation sign among variables with complete correlation. As specified in \cite{Langford}, if variables in both pairs $(x, y)$ and $(y, z)$ exhibit either positive or negative complete correlation, then variables $x$ and $z$ share positive complete correlation. If otherwise, $x$ and $z$ have negative complete correlation.
By the way, with correlation coefficient, the \textbf{symmetry} comes naturally by definition.

\begin{lemma}
\label{le1}
When variables\textemdash regardless of the sign\textemdash have complete correlation pairwise, the correlation matrix can be determined by just a single row.
\end{lemma}

\begin{proof}
We need to determine just signs since the absolute values of all entries are 1. We suppose $x_{1}$, $x_{2}$, ..., $x_{n}$ as random variables ($n$ $\in \mathbb{N}$). On the t-th row of the correlation matrix, we have the signs of the correlation coefficients between $x_{t}$ and each of $\{x_{i}\}_{i=1}^{n}$, so according to the aforementioned concepts of transitivity and symmetry, we can determine the signs of correlation coefficient between all pairs of $x_{i}$ and $x_{j}$ where $1\leq i,j\leq n$
\end{proof}

\begin{theorem}
\label{th1}
For every $n$ in $\mathbb{N}$, n-variable FC-datasets fall into 2$^{n-1}$  different correlation matrices. 
\end{theorem}

\begin{proof}
Each n-variable FC-datasets has its own correlation matrix. According to \ref{le1}, we can take the first row as the unique representation of the whole matrix, so We need to determine all cases of the first row based on minus or plus sign since the correlation matrix is all-ones pertinent to absolute values of entries. Given the positive correlation of a variable with itself, the first entry in the first row of the correlation matrix has a positive sign, for the rest of the entries on the row, we have two options (plus and minus) for each, so there are 2$^{n-1}$ different first rows for a correlation matrix, and consequently there are 2$^{n-1}$ correlation matrices for all n-variable FC-datasets. 
\end{proof}

\begin{corollary}
We have 2$^{n-1}$ FC-datasets for $n$ in $\mathbb{N}$.  
\end{corollary}  

\subsection{Fully-Uncorrelated Dataset}
\label{Fully-Uncorrelated Dataset}
\begin{definition}[\textbf{FU-dataset}]
The fully-uncorrelated dataset, denoted by \textbf{FU-dataset}, marks the other end of the spectrum mentioned earlier, where all variables never correlate with each other pairwise. 
\end{definition}

In the following, we have such a dataset with three variables: $A$, $B$, and $C$. 

\par
\begin{center}
\begin{minipage}{0.45\textwidth}
\includegraphics[width=\linewidth]{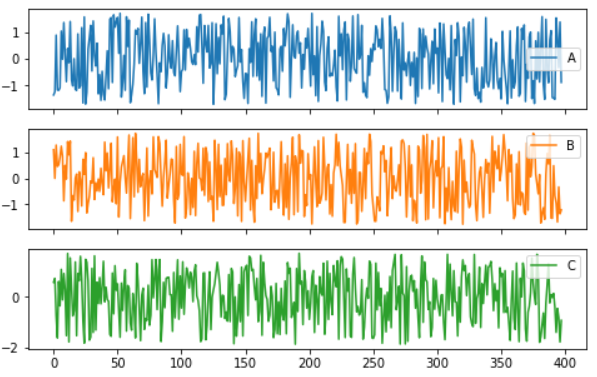}
\captionof{figure}{A Three-variable FU-dataset}
\label{fig2}
\end{minipage}\hfill
\end{center}

Unlike the fully-correlated case, we just have one FU-dataset for every $n$ $\in \mathbb{N}$. Based on the definition, the correlation matrices of FU-datasets is the identity matrix.

\section{The Behavior of Maximal Eigenvalue in FC-datasets}

The study of the correlation matrix of FC-datasets provides a matrix-based picture of the relationship between variables. Clearly, we observe the distribution of the maximal eigenvalue of the correlation matrix as FC-datasets' row grows one by one. However, before going into this distribution, we need to cite and prove an important theorem:

\begin{theorem}
\label{th2}
For every $n$ in $\mathbb{N}$, the correlation matrices of all 2$^{n-1}$ FC-datasets with n variables have the same set of eigenvalues.
\end{theorem}

\begin{proof}
 Consider an FC-dataset of $n$ random variables ($n$ in $\mathbb{N}$) whose $m$ variables ($m<n$) have complete negative correlation with the rest $n-m$ variables. We call this dataset $G$. On the other hand, we have a unique FC-dataset whose all $n$ variables have complete positive correlation with each other, and its correlation matrix is all-ones; we denote this unique FC-dataset by $H$. Now, consider $\lambda$ as an arbitrary eigenvalue for the correlation matrix of $H$, and $v=(v_{1},v_{2},...,v_{n})$ as an eigenvector for $\lambda$. According to \ref{le1} and concerning the first row of the correlation matrix of $G$, if we multiply -1 by those elements of $v$ that are placed at the same indices that -1s are posited on $G$'s correlation matrix's first row, we will obtain a new eigenvector, $v^\prime$, that makes $\lambda$ to be an eigenvalue for the correlation matrix of $G$. Without sacrificing the generality of the solution, we juxtapose the mentioned $m$ variables so that in the correlation matrix -1s lie next to each other.

\begin{align*}
\underbrace{\textbf{\Large +1}_{n\times n}}_{\text{Corr}(H)} \times
\underbrace{
\begin{bmatrix}
    v_1\\
    v_2\\
    \vdots\\
    v_n
\end{bmatrix}}_{v}   
&= 
\begin{bmatrix}
    \sum_{i=1}^{n}v_i\\
    \sum_{i=1}^{n}v_i\\
    \vdots\\
    \sum_{i=1}^{n}v_i
\end{bmatrix} 
=
\lambda
\times
\begin{bmatrix}
    v_1\\
    v_2\\
    \vdots\\
    v_n
\end{bmatrix} \\
\end{align*}

\begin{align*}
\underbrace{
\begin{bmatrix}
    \textbf{\large +1}_{(n-m)\times (n-m)} & \textbf{\large -1}_{(n-m)\times m} \\ 
    \textbf{\large -1}_{m\times (n-m)} & \textbf{\large +1}_{m\times m}
\end{bmatrix}}_{\textbf{Corr}(G)}
\times
\underbrace{
\begin{bmatrix}
    v_1\\
    \vdots\\
    v_{n-m}\\
    -v_{n-m+1}\\
    \vdots\\
    -v_n
\end{bmatrix}}_{v'}   
&= 
\begin{bmatrix}
    \sum_{i=1}^{n}v_i\\
    \vdots\\
    \sum_{i=1}^{n}v_i\\
    -\sum_{i=1}^{n}v_i\\
    \vdots\\
    -\sum_{i=1}^{n}v_i
\end{bmatrix} 
=
\begin{bmatrix}
    \lambda \times v_1\\
    \vdots\\
    \lambda \times v_{n-m}\\
    -\lambda \times v_{n-m+1}\\
    \vdots\\
    -\lambda \times v_n
\end{bmatrix} 
=
\\
\lambda
\times
\underbrace{
\begin{bmatrix}
    v_1\\
    \vdots\\
    v_{n-m}\\
    -v_{n-m+1}\\
    \vdots\\
    -v_n
\end{bmatrix}}_{v'}
\end{align*}
\end{proof}

This theorem holds that we can look at all 2$^{n-1}$ FC-datasets ($n$ in $\mathbb{N}$) as a single unique object. Mathematically speaking, all n-variable FC-datasets make a unique equivalence class under the relation of sharing the same set of eigenvalues for their correlation matrices each time samples add to the datasets. Thus, the behavior of the maximal eigenvalue for all FC-datasets is identical. It is because, in FC-datasets, variables entirely correlate with each other, and whenever the row grows, the same maximal eigenvalue comes up.

In the following, for the variables $A$, $B$, and $C$ mentioned earlier in \ref{Fully-Correlated Dataset} the computed maximal eigenvalues each time, as the row grows, are 3. The following shows the scatter plot and the histogram of the maximal eigenvalues.

\section{The Behavior of Maximal Eigenvalue in FU-datasets}

As with FC-datasets, we can apply the same procedure for FU-datasets. However, before delving into the account, we should restate that the correlation matrix of FU-datasets is identity matrix because non-trivial correlations are 0. All eigenvalues of identity matrices are 1, so the largest eigenvalue of the correlation matrix of an FU-dataset is 1.

Likewise, this statement asserts that given $n$ in $\mathbb{N}$, we can consider all FU-datasets with n variables as a unique dataset. From a mathematical standpoint, all FU-datasets belong to the same equivalence class under the relation of sharing the same set of eigenvalues for their correlation matrices, each time samples add to the datasets. So, they have the same maximal eigenvalue each time, and the behavior of this eigenvalue for all FU-datasets is identical. It is because, in FU-datasets, variables are nowhere near correlated to each other, and whenever the row grows, the same maximal eigenvalue is yielded.

If we apply the procedure on $A$,$B$, and $C$ in \ref{Fully-Uncorrelated Dataset}, we will achieve:

\par
\begin{center}
\begin{minipage}{0.8\textwidth}
\includegraphics[width=\linewidth]{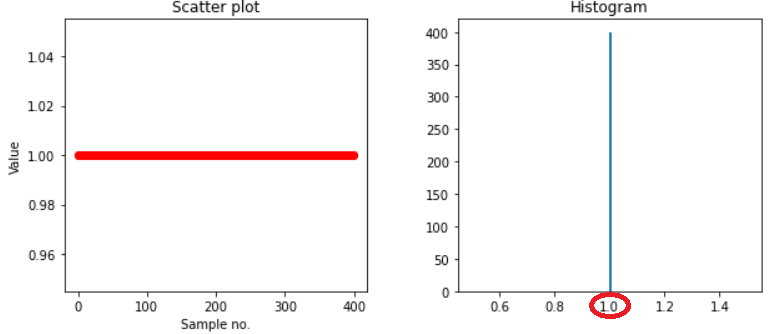}
\captionof*{figure}{}
\end{minipage}\hfill
\captionof{figure}{The Distribution of Maximal Eigenvalues As FU-datasets' Row Increases One by One}
\label{fig4}
\end{center}

\section{The Behavior of Maximal Eigenvalue While Adding Noise}

Needless to say, adding noise to FC-datasets eventually culminates in FU-datasets. Now, we want to explore what would happen to the behavior of the maximal eigenvalue when noise is added to FC-datasets. Our aim through this approach is indicating that the mean of an n-variable's maximal eigenvalue distribution rolls between 1 and n. The following theorems contribute to gaining a better appreciation of this issue. Meanwhile, during the entire of this paper, noise addition is done by the function of "Randbetween" in the Excel software. 

\begin{theorem}
\label{th3}
The largest eigenvalue of a correlation matrix is greater than or equal to 1. 
\end{theorem}

\begin{proof}
If all eigenvalues of a correlation matrix are absolutely smaller than 1, then the sum of them, which is the trace of the matrix, does not equal to the sum of the ones (1s) located on the main diagonal of the matrix; rather, it is strictly less than that, which is a paradox. Hence, at least one of the eigenvalues is greater than or equal to 1.
\end{proof}

\begin{theorem}
\label{th4}
The largest possible eigenvalue for a correlation matrix is the number of variables, which is taken by the correlation matrices of FC-datasets. 
\end{theorem}

\begin{proof}
$\lambda_{max}\leq\sum_{i=1}^{n}{\lambda_{i}}=trace=\sum_{i=1}^{n}{1}=n$. $n$ is an eigenvalue for the correlation matrices of FC-datasets because the first column of the correlation matrices is its non-trivial eigenvector. 
\end{proof}

These theorems assert that the distribution of the maximal eigenvalue in a dataset with n variables typically falls between that of its fully uncorrelated (FU) and fully correlated (FC) counterparts. Specifically, the maximal eigenvalues lie within the closed interval of $[1, n]$.

In the following, we pick the first FC-dataset of (a) in \ref{Fully-Correlated Dataset} and we add noise to it. In (a), the pairwise correlations are positive and complete; then, we attain the distribution of the maximal eigenvalues when its row increases one by one.

\par
\begin{center}
\begin{minipage}{0.45\textwidth}
\includegraphics[width=\linewidth]{1.PNG}
\captionof*{figure}{(a)}
\end{minipage}\hfill
\begin{minipage}{0.55\textwidth}
\includegraphics[width=\linewidth]{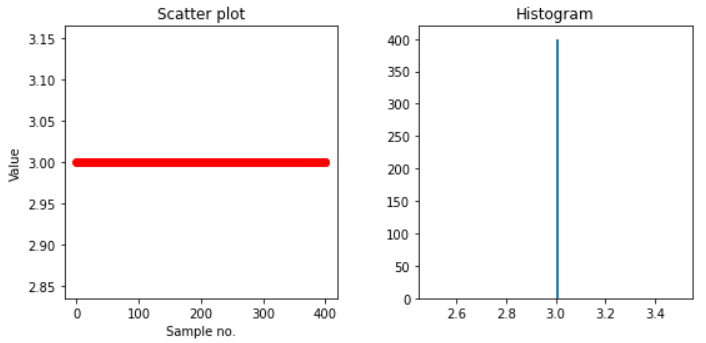} 
\captionof*{figure}{(b)}   
\end{minipage}
\begin{minipage}{0.45\textwidth}
\includegraphics[width=\linewidth]{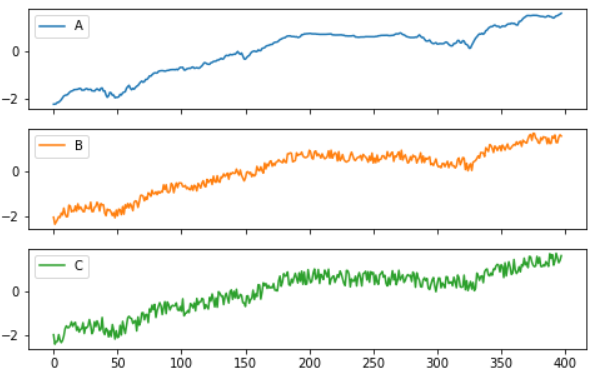}
\captionof*{figure}{(c)}
\end{minipage}\hfill
\begin{minipage}{0.55\textwidth}
\includegraphics[width=\linewidth]{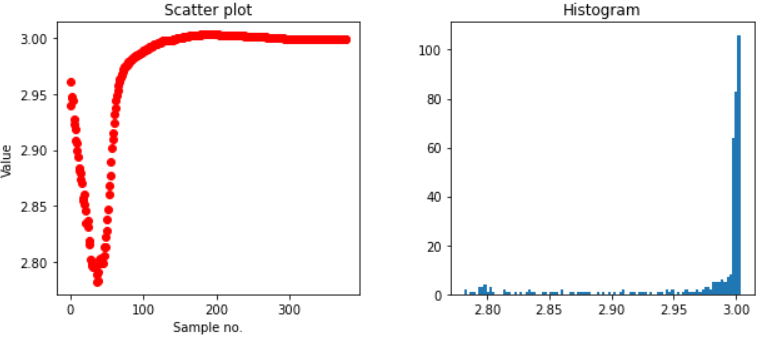} 
\captionof*{figure}{(d)}   
\end{minipage}
\begin{minipage}{0.45\textwidth}
\includegraphics[width=\linewidth]{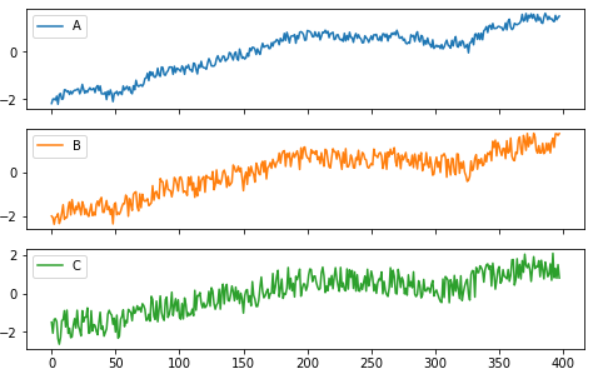}
\captionof*{figure}{(e)}
\end{minipage}\hfill
\begin{minipage}{0.55\textwidth}
\includegraphics[width=\linewidth]{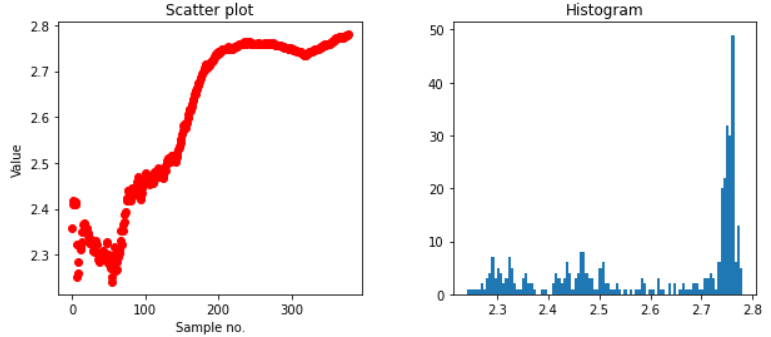}
\captionof*{figure}{(f)}
\end{minipage}\hfill
\begin{minipage}{0.45\textwidth}
\includegraphics[width=\linewidth]{9.PNG}
\captionof*{figure}{(e)}
\end{minipage}\hfill
\begin{minipage}{0.55\textwidth}
\includegraphics[width=\linewidth]{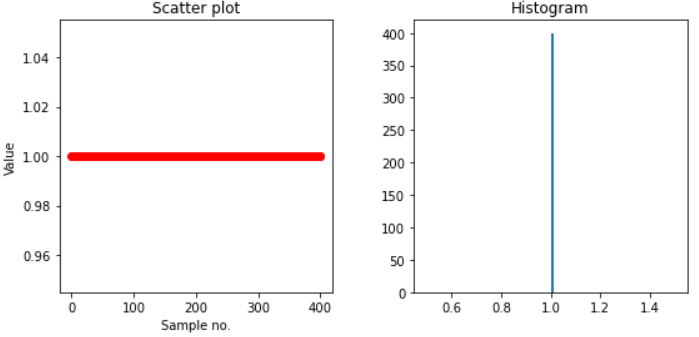}
\captionof*{figure}{(f)}
\end{minipage}\hfill
\captionof{figure}{The Distribution of Maximal Eigenvalues As the FC-dataset's Row Increases One by One In The Case of Adding Noise}
\label{fig5}
\end{center}

Theorems \ref{th3} and \ref{th4} as well as the preceding set of figures (\ref{fig5}) indicate that the distribution of a typical dataset with n variables lies between that of its equivalent FU-dataset and FC-dataset. Explicitly, the mean of the maximal-eigenvalue distribution of a typical dataset with n variables lies between the mean of the the same distribution for its FU-and-FC counterpart datasets, namely 1 and n. 

Since each time the sample value increases(new row of the matrix), the same maximal eigenvalue is achieved for FU-and-FC datasets, and both the means are the same maximal eigenvalues. For FU-datasets, it is 1, and for FC-datasets n.

\section{Multivariable Pearson's Correlation Coefficient}

In this section, we assert and subsequently demonstrate that as variables in a dataset lose their pairwise correlation, the mean of its maximal eigenvalue distribution tends to approach 1. Conversely, as pairwise correlations improve, the mean shifts toward n. Visual representations of these trends can be observed in the graphs presented in Figure \ref{fig5}. Prior to introducing the formula, we have outlined two theorems:

\begin{lemma}[Spectral Decomposition]
\label{le2}
Any correlation matrix is diagonalizable by a unitary matrix.
\end{lemma}

\begin{proof}
\cite{Zhang} indicates how symmetric matrices are diagonalizable by unitary matrices, and so are correlation matrices because of being symmetric.
\end{proof}

\begin{theorem}[Spectral Radius]
\label{th5}
If A is a correlation matrix, the Euclidean 2-norm of $A$ equals its maximal eigenvalue: $\| \bm{A} \|_{2}=\lambda_{max}$. For the sake of convenience, we will show $\| \bm{.} \|_{2}$ without subscription.
\end{theorem}

\begin{proof}
If $\lambda_{max}$ indicates the largest eigenvalue of A, for each unit eigenvector $x$ of $\lambda_{max}$, we will have $ \lambda_{max}=\mid \lambda_{max} \mid = \| \bm{\lambda_{max}x} \|=\| \bm{Ax} \|\leq \sup_{\| \bm{x} \|=1}{\| \bm{Ax} \|=\| \bm{A} \|}$. For the reverse inequality, we use the theorem \ref{le2}, so we will have the unitary matrix $U$ so that $A=U^{t}DU$, where D is a diagonal matrix with $A$'s eigenvalues on its main diagonal. For an arbitrary unit vector $v$, we have $\| \bm{Av} \|=\| \bm{U^{t}DUv}\|$. Now, since $\| \bm{Ax} \|= (x^{t}A^{t}Ax)^{\frac{1}{2}}$ for an arbitrary $A$ and $x$, and since the norm of a unitary matrix is 1, so we will have $\| \bm{U^{t}DUv}\|\leq \| \bm{U^{t}}\|\| \bm{DUv}\|\leq \| \bm{U^{t}}\|\| \bm{D}\|\| \bm{Uv}\|\leq \| \bm{U^{t}}\|\| \bm{D}\|\| \bm{U}\|\| \bm{v}\|=  \| \bm{D}\|\leq\lambda_{max}$. Hence, taking supremum on all unit vectors shows $\| \bm{A} \|\leq\lambda_{max}$.
\end{proof}

We have provided a mathematical foundation for the desired indicator. We, actually, proved that correlation matrices are diagonalizable, and their Euclidean norm\_2 is synonymous with its largest eigenvalue. Norm\_2 is a strictly increasing function of absolute values of correlation matrix's entries because it is a combination of the square root function and the quadratic function, which are  both increasing on the non-negative set of real numbers. All correlation matrices have 1s on their main diagonal, then norm\_2 takes its minimum on the correlation matrices of FU-datasets since their off-diagonal entries are all zero. Its norm equals
\begin{equation}\label{first_equation}
    \sqrt{\sum_{i=1}^{n}{\sum_{j=1}^{n}{\delta_{ij}}}} = \sqrt{n}
\end{equation}

Where $\delta_{ij}$ is Kronecker delta and $n$ is the number of variables. The norm also takes its maximum on correlation matrices of FC-datasets because the absolute value of all their entries is 1, which is the utmost. Its norm equals
\begin{equation}\label{second_equation}
    \sqrt{\sum_{i = 1}^{n}{\sum_{j=1}^{n}{(\pm 1)}}}=n
\end{equation}

For a given dataset, the norm of its correlation matrix falls within the range defined by the norm of its fully uncorrelated (FU) dataset's correlation matrix and that of its fully correlated (FC) dataset's correlation matrix. Furthermore, when comparing the correlation matrix of the dataset to those of the FU and FC datasets, the discrepancies in their norms indicate the differences in absolute values of their respective entries. The equivalence established between the norm and the maximal eigenvalue in Theorem \ref{th5} underscores that comparing the largest eigenvalues serves as a genuine indicator of the degree to which a dataset is either fully correlated or fully uncorrelated. Specifically, the distance of the mean maximal eigenvalue from 1 reflects the general correlation level among variables, while its distance from n signifies the overall lack of correlation among variables.

Now in light of what we have discussed, we present the formula of Pearson's correlation coefficient for $n$ random variables:
\begin{equation}\label{third_equation}
    \rho_{x_{1}, x_{2}, ..., x_{n}}=\frac{mean({\lambda_{max}})-1}{n-1}
\end{equation}

Given the theorem \ref{th5}, we can rewrite the preceding formula based on the distribution of the correlation matrix norm:
\begin{equation}\label{fourth_equation}
    \rho_{x_{1}, x_{2}, ..., x_{n}}=\frac{mean(\| \bm{.} \|_{2})-\sqrt{n}}{n-\sqrt{n}}
\end{equation}

The last version seems much better because the computation of norm is less accompanied by error, unlike that of the maximal eigenvalue. On top of that, its computation is not as complicated and as time-consuming as it is with maximal eigenvalue.

According to theorems \ref{th2},\ref{th3}, \ref{th4}, \ref{th5}, $\rho$ lies in the interval of [0,1]: in the FC-case we obtain 1, and in the FU-case we attain 0; the $\rho_{x_{1}, x_{2}, ..., x_{n}}$ of every other dataset lies between 0 and 1. If variables in the dataset are less correlated on the whole, Its $\rho_{x_{1}, x_{2}, ..., x_{n}}$ tends to 0; if otherwise, tending to 1.

\noindent
It is worth mentioning that the computation of $\rho$ in this case is a bit conservative, as we take average each time a new sample adds, while 
\begin{equation}\label{fifth_equation}
    \rho_{x_{1}, x_{2}, ..., x_{n}}=\frac{\| \bm{.} \|_{2}-\sqrt{n}}{n-\sqrt{n}}
\end{equation}

is way closer to the the common notion of Pearson's correlation for two variables.

\section{Feature Selection}
\label{Feature Selection}

In tasks like supervised classification, where the goal is prediction, multivariate correlation plays a crucial role in selecting appropriate variables. When visualizing data points in a coordinated plane with predictors as axes, the accuracy of our predictions is inversely affected by the mixing of labeled data points. As labeled data becomes more intertwined, algorithms, to prevent overfitting, draw a more generalized separator that avoids precisely fitting a high-dimensional surface to divide data groups with different labels. While this strategy aims to enhance generalization, it comes at the cost of training accuracy, resulting in lower test accuracy. To mitigate this, careful consideration should be given to selecting a subset of predicting variables. The following illustration exemplifies the essence of this idea.

\par
\begin{center}
\begin{minipage}{0.40\textwidth}
\includegraphics[width=\linewidth]{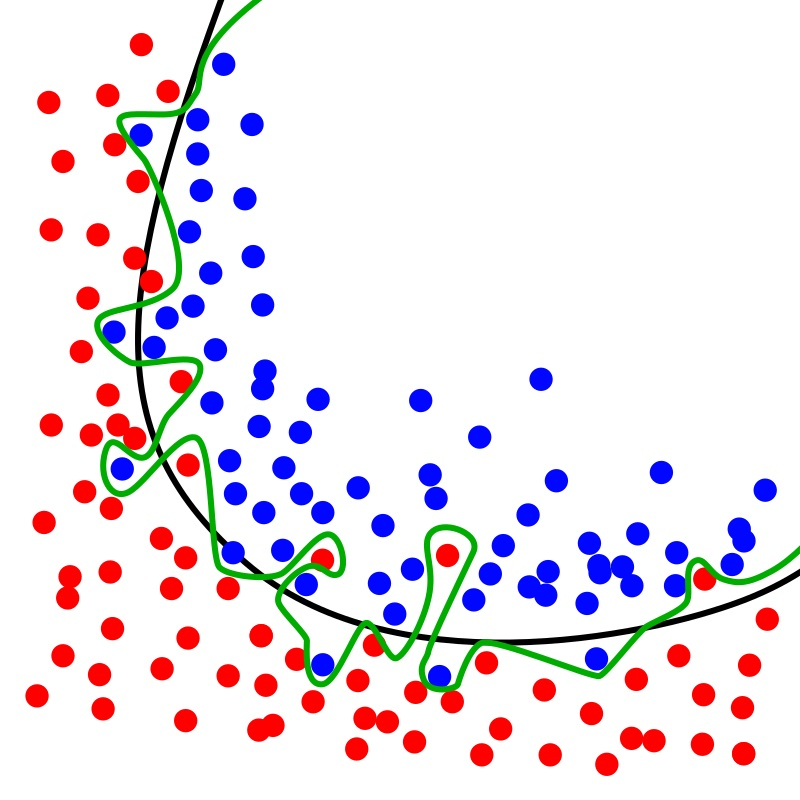}
%\captionof*{figure}{Optimal fitting}
\end{minipage}\hfill
\captionof{figure}{The Green Wavy Line Is Over-Fitting Separator, And The Black Curved One Is The Optimal-Fitting Separator}
\label{fig6}
\end{center}

Indeed, the more the total curvature of the over-fitting separator, the more the difference between the curvatures of optimal-fitting and over-fitting separators, resulting in lower accuracy.

Now, we narrow our focus down to fully-correlated datasets in $\ref{Fully-Correlated Dataset}$, in order to predict variable A while having B and C. The B-C plane will be one of the followings: 

\par
\begin{center}
\begin{minipage}{0.50\textwidth}
\includegraphics[width=\linewidth]{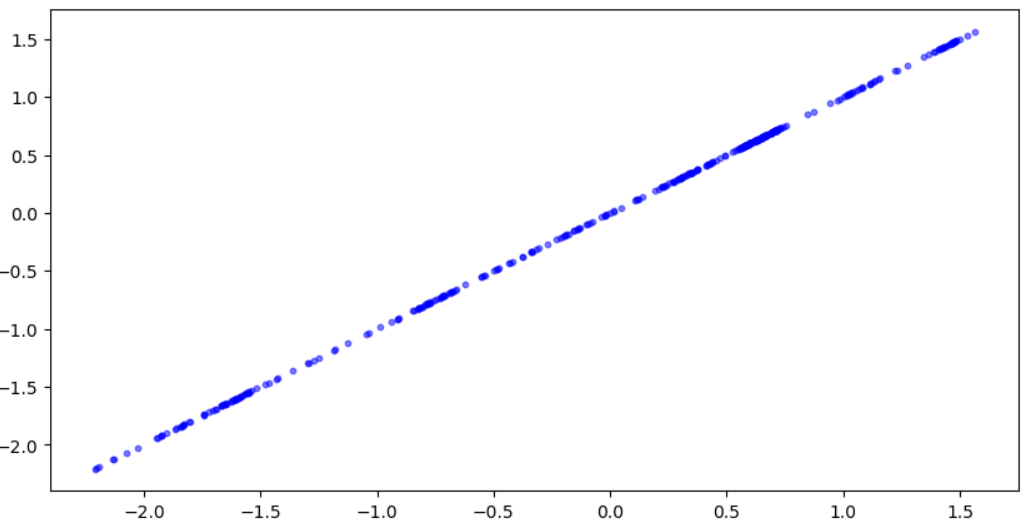}
\captionof*{figure}{(a)}
\end{minipage}\hfill
\begin{minipage}{0.50\textwidth}
\includegraphics[width=\linewidth]{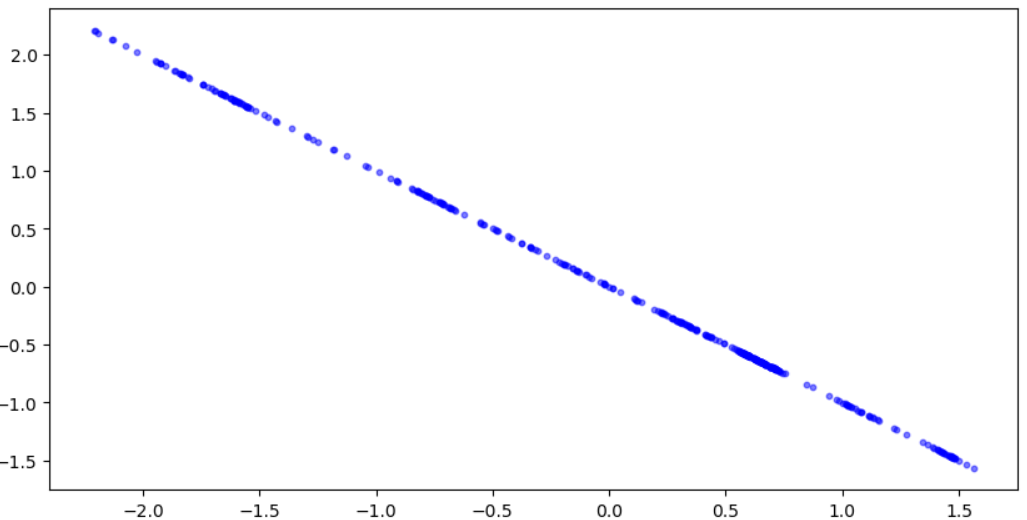} 
\captionof*{figure}{(b)}   
\end{minipage}
\captionof{figure}{The B-C Planes of Fully-Correlated Datasets}
\label{fig7}
\end{center}

We label data points on these planes based on whether they belong to the upper-median A or to the lower-median A. Data points corresponding lower-median A are blue, and those belong to upper-median A red. This way, we have a balanced labeling, which makes us sure the number of blues and reds are equal.

\par
\begin{center}
\begin{minipage}{0.50\textwidth}
\includegraphics[width=\linewidth]{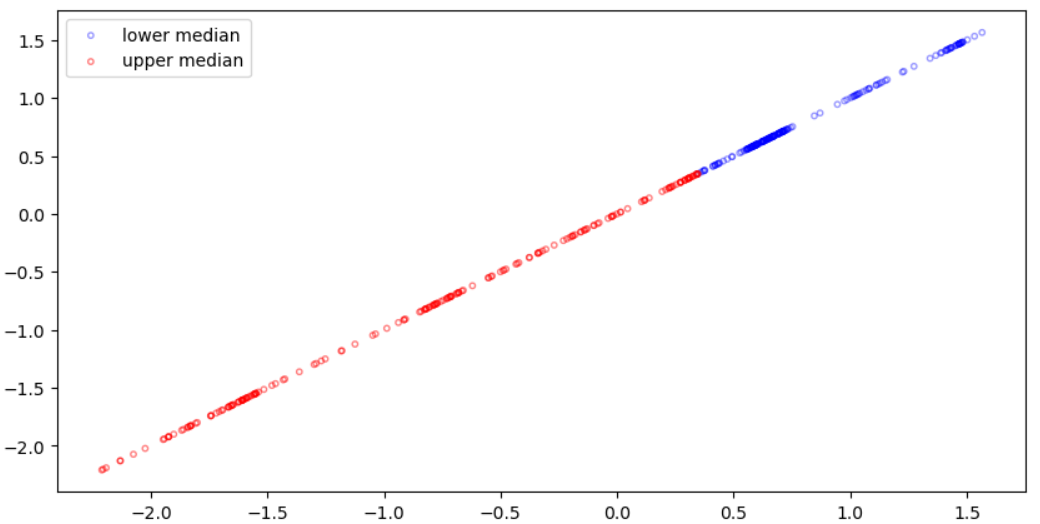}
\captionof*{figure}{(a)}
\end{minipage}\hfill
\begin{minipage}{0.50\textwidth}
\includegraphics[width=\linewidth]{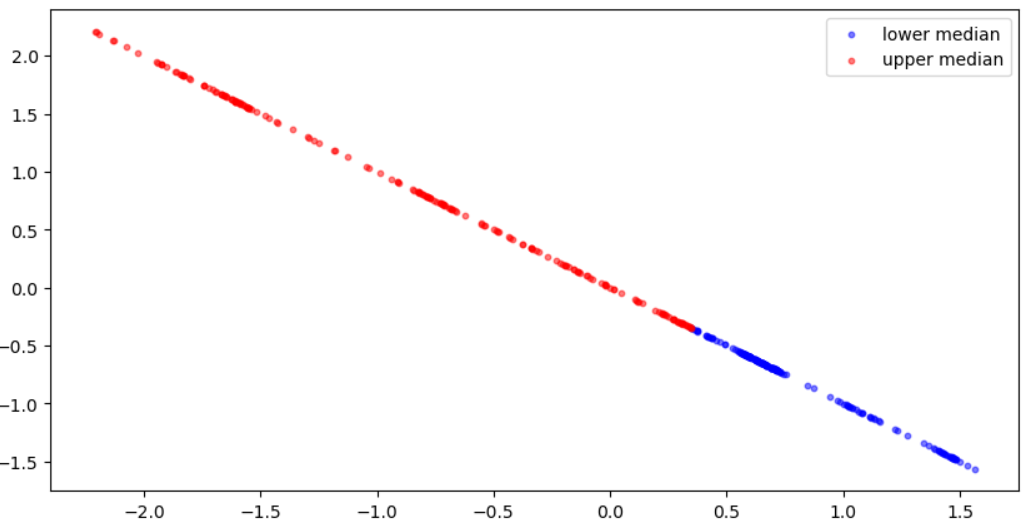} 
\captionof*{figure}{(b)}   
\end{minipage}
\begin{minipage}{0.496\textwidth}
\includegraphics[width=\linewidth]{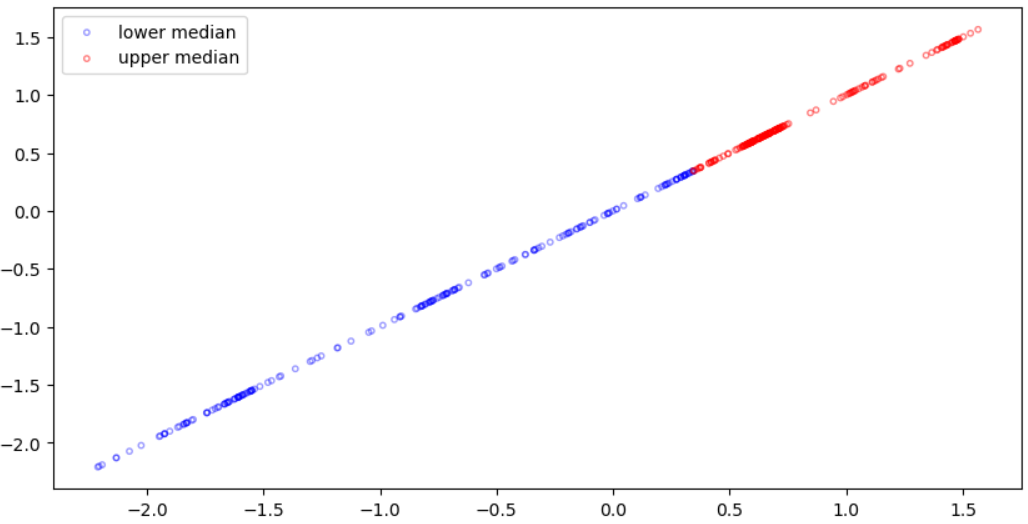} 
\captionof*{figure}{(c)}   
\end{minipage}
\begin{minipage}{0.496\textwidth}
\includegraphics[width=\linewidth]{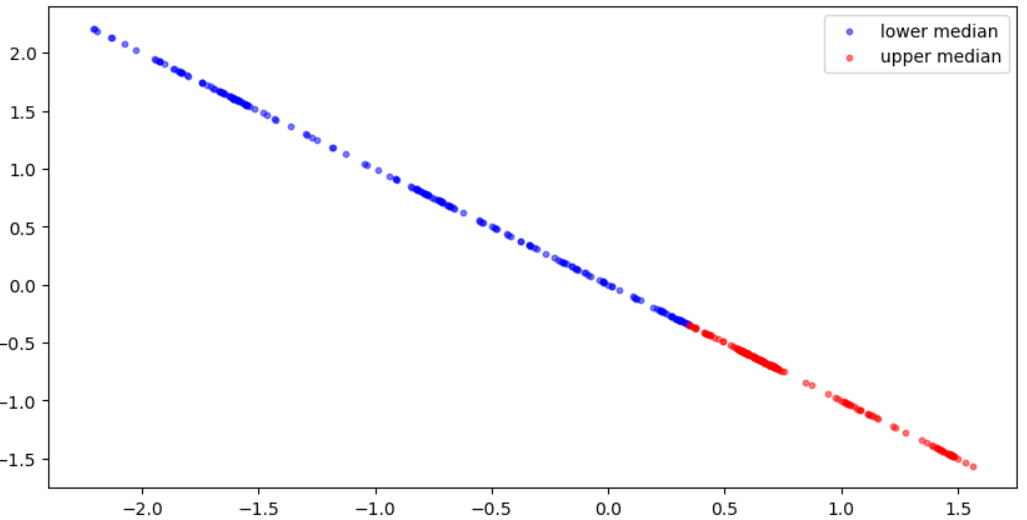} 
\captionof*{figure}{(d)}   
\end{minipage}
\captionof{figure}{The Labeled Data in B-C Planes of Fully-Correlated Datasets}
\label{fig8}
\end{center}

In each of preceding situations, a single dot on the line can separate sets of data, and there is no over-fitting challenge at all.  In the rest of this paper we just consider dataset (a) corresponding figure (a) in $\ref{Fully-Correlated Dataset}$.

We add noise to this dataset to examine the situation when pairwise correlations diminish and each time compute $\rho$.

\par
\begin{center}
\begin{minipage}{0.45\textwidth}
\includegraphics[width=\linewidth]{1.PNG}
%\captionof*{figure}{(a)}
\end{minipage}\hfill
\begin{minipage}{0.55\textwidth}
\includegraphics[width=\linewidth]{22.PNG} 
%\captionof*{figure}{(b)}   
\end{minipage}
\captionof{figure}{$\rho_{A, B, C}=1$}

\begin{minipage}{0.45\textwidth}
\includegraphics[width=\linewidth]{16.PNG}
%\captionof*{figure}{(c)}
\end{minipage}\hfill
\begin{minipage}{0.55\textwidth}
\includegraphics[width=\linewidth]{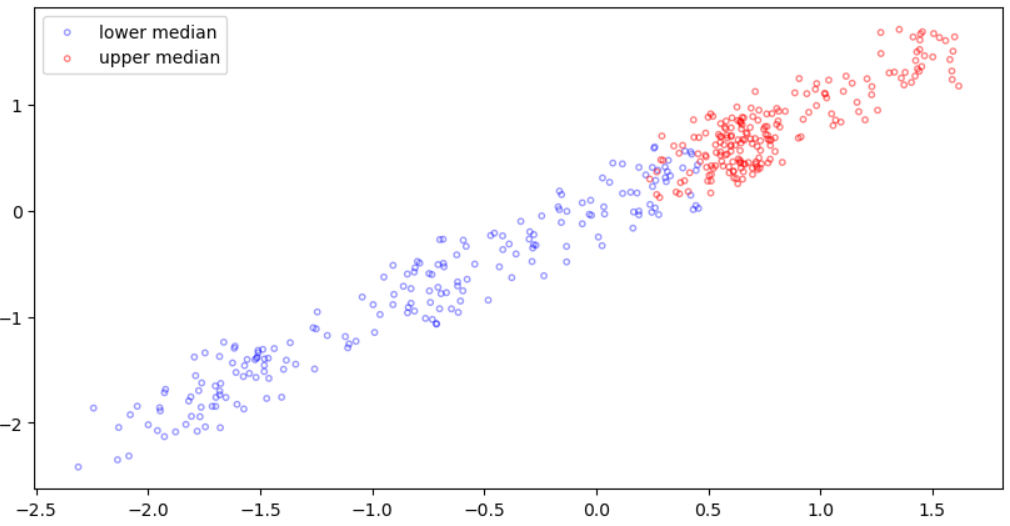} 
%\captionof*{figure}{(d)}   
\end{minipage}
\captionof{figure}{$\rho_{A, B, C}=0.9768$}

\begin{minipage}{0.45\textwidth}
\includegraphics[width=\linewidth]{17.PNG}
%\captionof*{figure}{(e)}
\end{minipage}\hfill
\begin{minipage}{0.55\textwidth}
\includegraphics[width=\linewidth]{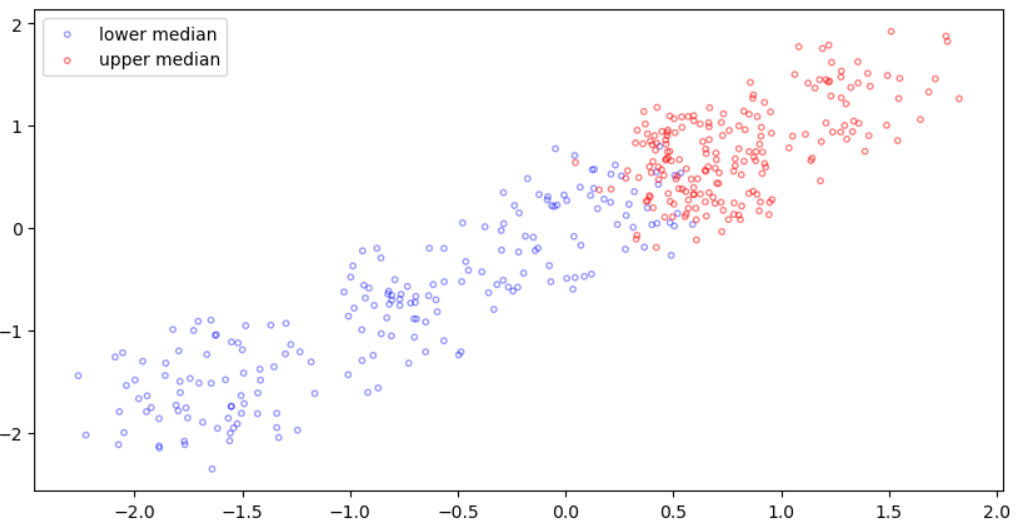}
%\captionof*{figure}{(f)}
\end{minipage}\hfill
\captionof{figure}{$\rho_{A, B, C}=0.9196$}

\begin{minipage}{0.45\textwidth}
\includegraphics[width=\linewidth]{9.PNG}
%\captionof*{figure}{(e)}
\end{minipage}\hfill
\begin{minipage}{0.55\textwidth}
\includegraphics[width=\linewidth]{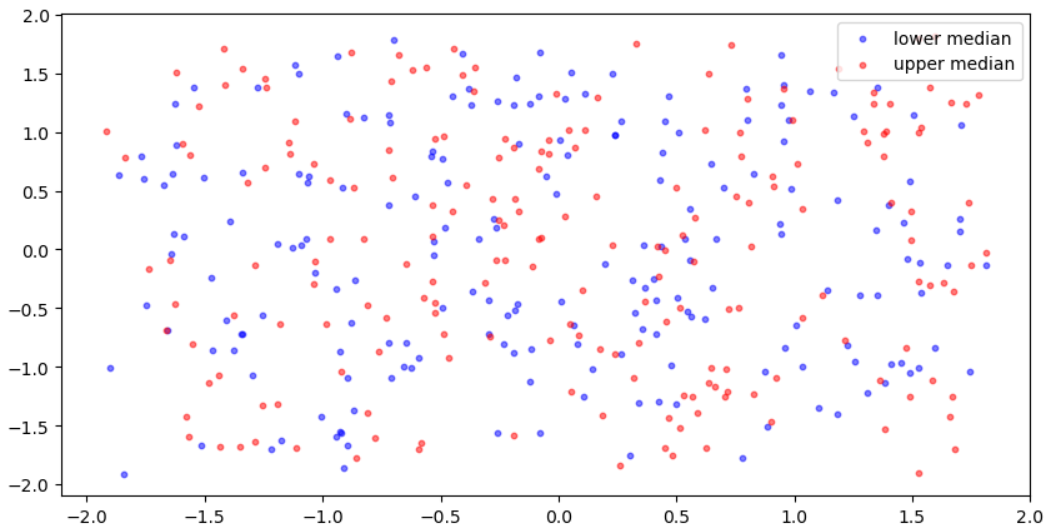}
%\captionof*{figure}{(f)}
\end{minipage}\hfill
\captionof{figure}{$\rho_{A, B, C}=0.0019$}

\label{fig9}
\end{center}

For each of the preceding datasets, we have computed the multivariate correlation that are displayed in captions. When we add noise, one thing leads to another in a domino effect and makes the prediction accuracy lower. In fact, with the loss of pairwise correlations, the data points are become mixed together more, which goes hand in hand with decrease in the multivariate correlation. Then, the curvature of over-fitting separator increases, leading to decrease in the final accuracy. Indeed, there is an inverse continuous relationship between the multivariate correlation and the curvature of over-fitting separator.  

It is worth mentioning that adding noise does not make the situation special, and not make our conclusion biased. Actually, we add noise to lose pairwise correlation to be able to examine a various range of situations. 

\subsection{Measuring Noise}
\label{Measuring Noise}

In regard with the potential noise in datasets, two sources of noise are recognizable: the first one is the noise coming from predictors, and the second one, noise stemming from the supervising variable that we call it labeling noise for the following reason. Balanced labeling is systematic, for example based on measures of central tendency. How data points in a coordinated plane are labeled entirely depends upon how values of the target variable are correlated in correspondence with those of the predicting variables.
In better words, it relies on the correlations of the target variable with predicting ones.

Consider n variables of $x_{1}, x_{2}, ..., x_{n}$ from which the first $n-1$ are going to predict $x_{n}$ through classification. The first type of noise that shows how scattered data points in the plane whose axes are $x_{1}, x_{2}, ..., x_{n}$ can be computed via $\rho_{x_{1}, x_{2}, ..., x_{n-1}}$.
When the noise coming from predators is not meaningfully large, the labeling noise is quantifiable through the following relation
\begin{equation}\label{sixth_equation}
    max(0, ((1-\rho_{x_{1}, x_{2}, ..., x_{n}})-(1-\rho_{x_{1}, x_{2}, ..., x_{n-1}})))
\end{equation}

Clearly, we remove the noise of predictors from $1-\rho_{x_{1}, x_{2}, ..., x_{n}}$. It goes without saying that we can readily compare labeling noise in different datasets with different variables, as $\rho$ is adjusted by the number of predictors, naturally by its definition. In the following, we are going to add noise only to the supervising variable of A to see how the noise of labeling increases. The labeling process is balanced: upper-median A and lower-median A, which remains the same, so do B and C.

\par
\begin{center}
\begin{minipage}{0.45\textwidth}
\includegraphics[width=\linewidth]{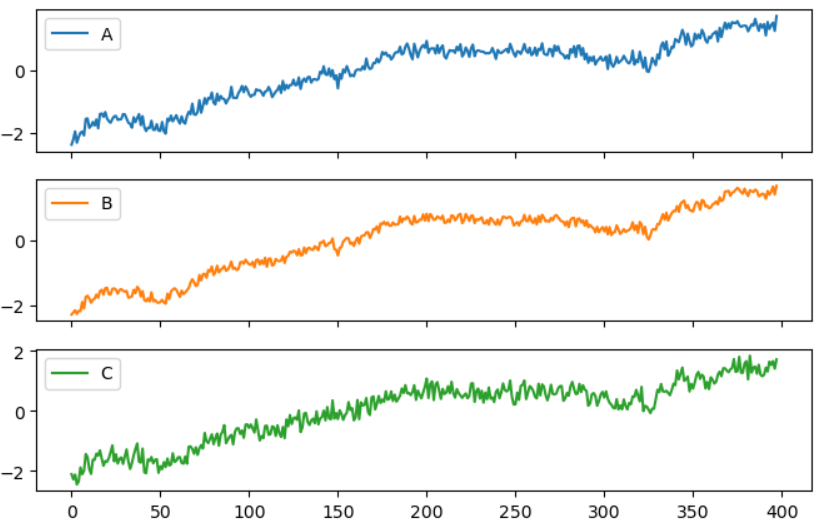}
%\captionof*{figure}{(a)}
\end{minipage}\hfill
\begin{minipage}{0.55\textwidth}
\includegraphics[width=\linewidth]{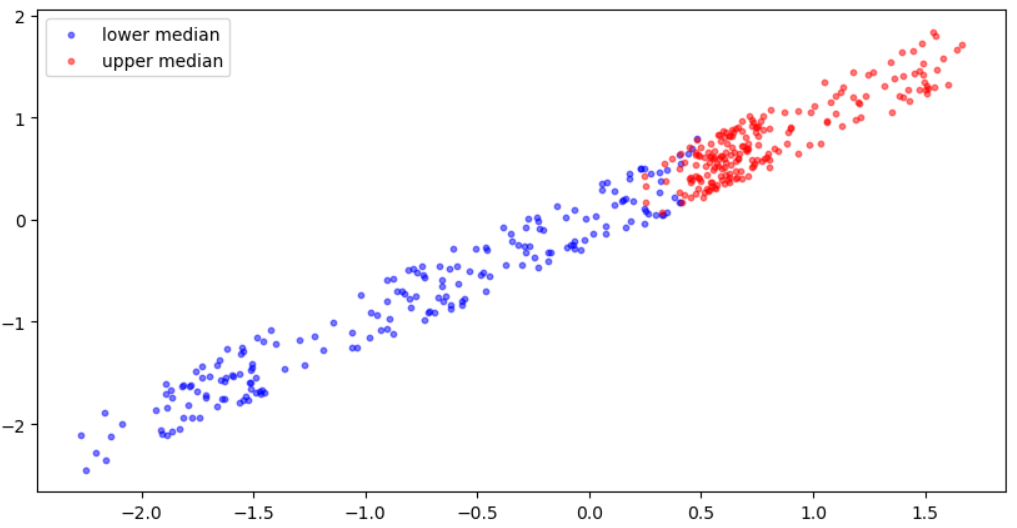} 
%\captionof*{figure}{(b)}   
\end{minipage}
\captionof{figure}{$1-\rho_{A, B, C}=0.0213$, $1-\rho_{B, C}=0.0278$; The Noise of Labeling is 0}

\begin{minipage}{0.45\textwidth}
\includegraphics[width=\linewidth]{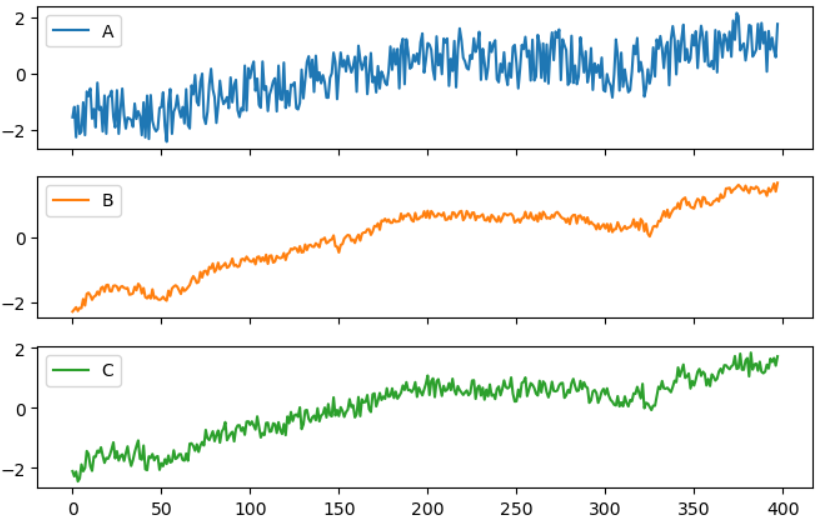}
%\captionof*{figure}{(c)}
\end{minipage}\hfill
\begin{minipage}{0.55\textwidth}
\includegraphics[width=\linewidth]{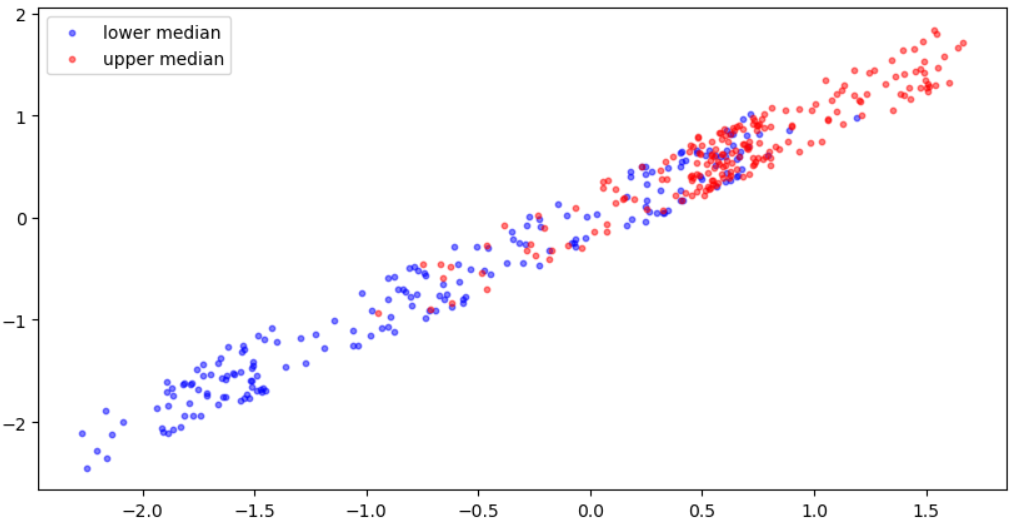} 
%\captionof*{figure}{(d)}   
\end{minipage}
\captionof{figure}{$1-\rho_{A, B, C}=0.1962$, $1-\rho_{B, C}=0.0278$; The Noise of Labeling is 0.1683}

\begin{minipage}{0.45\textwidth}
\includegraphics[width=\linewidth]{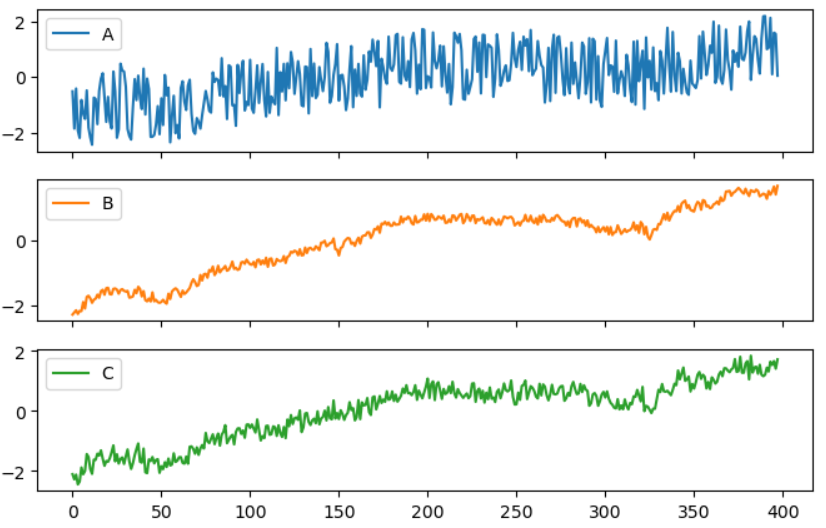}
%\captionof*{figure}{(e)}
\end{minipage}\hfill
\begin{minipage}{0.55\textwidth}
\includegraphics[width=\linewidth]{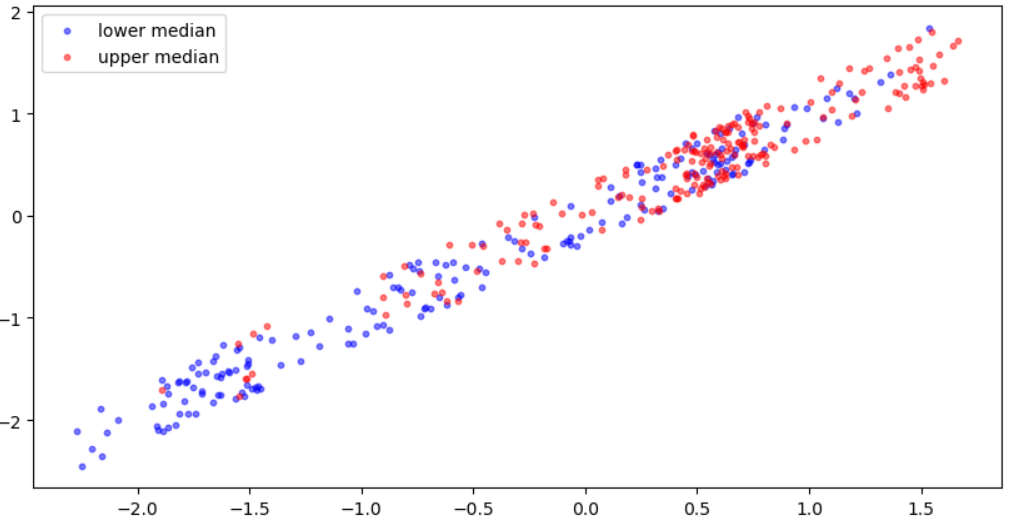}
%\captionof*{figure}{(f)}
\end{minipage}\hfill
\captionof{figure}{$1-\rho_{A, B, C}=0.3479$, $1-\rho_{B, C}=0.0278$; The Noise of Labeling is 0.3201}

\label{fig10}
\end{center}

\section{Conclusion}
\label{Conclusion}
We have presented a multivariable version of Pearson's correlation coefficient, which obviates limitations associated with common correlation coefficient. It is not just the formula itself, but how we achieved it. In fact, this article has mapped this measure of association into a broad context in order to extend it. Similarly, We can find correspondent matrix features of other statistical association measures and explore them with the aid of powerful tools that matrix theory offers. We hope this approach can contribute to a better appreciation of random matrix theory methods and tools and their importance in expanding and ameliorating statistical association measures, or even vice versa. 

By the way, we showed that, in supervised learning, with the aid of the multivariable correlation coefficient, we can separate noise of data into two different categories: noise of labeling and noise of predictors. Actually, with this method, finance experts can gauge optimally the potential noise of each source when predicting a variable like stock price or inflation rate, thereby settling for more suitable dataset to forecast.

\bibliographystyle{chicago}

\bibliography{Bibliography-MM-MC}

\end{document}